\documentclass[11pt,reqno]{amsart}
\usepackage{amsfonts, amsmath, amssymb, color}
\usepackage{amsthm}
\usepackage{booktabs}
\usepackage{enumitem}
\usepackage{float}
\usepackage{hhline}
\usepackage{soul, color, xcolor}
\usepackage{lipsum}
\usepackage{lscape}
\usepackage{subfig}
\usepackage{textcomp}
\usepackage{tikz}
\usetikzlibrary{decorations.pathmorphing,shapes,arrows,positioning}
\usepackage{xcolor}
\usepackage{young}
\usepackage{youngtab}
\usepackage{ytableau}
\usepackage{rotating}
\usepackage{graphicx}
\usepackage[colorlinks=true, linkcolor=blue, citecolor=blue, urlcolor=blue,
]{hyperref} 

\hypersetup{colorlinks,linkcolor=blue,urlcolor=blue,citecolor=blue}

\usepackage{scalefnt}

\usepackage{scalefnt}
\usetikzlibrary{calc}
\usepackage{a4wide}

\usepackage[numbers,sort&compress]{natbib}

\allowdisplaybreaks[4]
\theoremstyle{plain}
\newtheorem{thm}{Theorem}[section]
\newtheorem{lem}[thm]{Lemma}
\newtheorem{prop}[thm]{Proposition}

\newtheorem{rem}[thm]{Remark}

\theoremstyle{definition}

\newtheorem{exmp}[thm]{Example}

\numberwithin{equation}{section} \errorcontextlines=0

\makeatletter

\newcommand{\Rmnum}[1]{\expandafter\@slowromancap\romannumeral #1@}
\makeatother

\begin{document}
\title{Tighter bounds of generalized monogamy and polygamy relations}
\author{Yue Cao}
\address{School of Mathematics, South China University of Technology,
Guangzhou, Guangdong 510640, China}
\email{434406296@qq.com}
\author{Naihuan Jing*}
\address{Department of Mathematics, North Carolina State University, Raleigh, NC 27695, USA}
\email{jing@ncsu.edu}
\author{Yiling Wang}
\address{Department of Mathematics, North Carolina State University, Raleigh, NC 27695, USA}
\email{ywang327@ncsu.edu}
\subjclass[2010]{Primary: 85; Secondary:}
\thanks{$*$Corresponding author: jing@ncsu.edu}
\begin{abstract}
 We study generalized monogamy and polygamy relations for concurrence of assistance and negativity of assistance
 using parametrized bounds in general multi-partite quantum systems. The new method overcomes the shortcomings of previous studies where a method is only good at a particular region. We provide detailed examples to show why the new approach is effective in all regions.
\end{abstract}
\maketitle

\section{\textbf{Introduction}}

Quantum information theory has made significant strides in recent decades, largely due to the fascinating properties of quantum entanglement \cite{CAF,HHHH,DFSC}. This distinctly quantum phenomenon defies classical intuition by establishing correlations between particles that are independent of their spatial separation. These correlations are fundamental to quantum technologies, such as quantum computing, quantum cryptography, and quantum communication. A deep understanding of how entanglement can be distributed and controlled is essential for advancing these technologies.

A particularly intriguing aspect of entanglement in multipartite quantum systems is its constrained nature, which is encapsulated in the concepts of monogamy and polygamy. Monogamy implies that if two subsystems are highly entangled, their ability to share entanglement with additional subsystems is limited \cite{ASI,P,AMG}. This principle, formally introduced by Coffman, Kundu, and Wootters \cite{CKW}, has become a cornerstone in the study of quantum correlations:
$C^{2}(\rho_{A|BC})\geq C^{2}(\rho_{AB})+C^{2}(\rho_{AC})
$, where $\rho_{AB}$ and $\rho_{AC}$ are the reduced density matrices.
Conversely, polygamy describes scenarios where entanglement can be shared more freely across multiple subsystems, providing a complementary perspective to monogamy \cite{OV, G,JLLF,ZF1,KPSS,KDS,OF}. Gour et al \cite{GMS} established the first polygamy inequality by using the squared concurrence of assistance (CoA):
$
C^{2}(|\psi\rangle_{A|BC})\leq C^{2}_{a}(\rho_{AB})+C^{2}_{a}(\rho_{AC}),
$
where $\rho_{AB}$ and $\rho_{AC}$ are the reduced density matrices of $|\psi\rangle_{ABC}\langle\psi|$.

Extensive research has explored the mathematical formulation of these principles, particularly through entanglement measures such as concurrence and negativity 
\cite{CK, OF, LCO}. Concurrence, initially developed to quantify entanglement between pairs of qubits, has been extended to investigate entanglement distribution in larger multipartite systems. Over time, generalized monogamy and polygamy inequalities \cite{K1, K2, CJW, ZF2,YCLZ,JF,JFL, CJMW} based on concurrence have been formulated, offering insights into entanglement behavior across complex quantum systems. These inequalities are crucial for quantum information security, where precise control over entanglement distribution is vital.

Despite these advancements, current monogamy and polygamy relations \cite{JFQ, ZJZ1,ZJZ2,ZJLM, SXZWF} are often not tight enough, meaning that the inequalities do not fully capture the constraints on entanglement distribution in certain scenarios. This lack of precision can lead to suboptimal results in practical applications, such as quantum cryptography \cite{GRTZ}, where the robustness of security protocols is closely linked to accurate entanglement characterization. Thus, there is a pressing need to refine these inequalities to provide tighter bounds that more accurately reflect quantum correlation behavior.

In this work, we address this challenge by proposing a new set of parameterized bounds that enhance the generalized monogamy and polygamy relations. By introducing a family of tighter bounds based on the $\alpha$th power of concurrence, we demonstrate a more refined understanding of entanglement distribution in N-qubit systems. Our results improve existing inequalities and offer greater flexibility, as the parameterization allows for a broad range of applications depending on the specific characteristics of the quantum system \cite{GRTZ}.

Additionally, we extend our analysis to the measure of negativity, which is particularly useful for studying entanglement in mixed quantum states. The connection between negativity and concurrence provides a natural framework for applying our new bounds to a wider array of quantum systems, ensuring applicability across different contexts, from pure to mixed states.

To illustrate the impact of our approach, we provide detailed comparisons between our new bounds and those in the current literature. Through explicit examples, we show that our parameterized bounds consistently outperform previous results \cite{CJW,ZF2,YCLZ,JF,JFL}, offering tighter constraints on both monogamy and polygamy relations. This advancement not only enhances theoretical understanding of quantum entanglement distribution but also has practical implications for developing more secure and efficient quantum information protocols.

The remainder of this paper is structured as follows: Section 2 outlines the mathematical tools and techniques used to derive our new bounds. Section 3 presents the tighter polygamy relations, followed by Section 4, which develops optimized bounds for monogamy relations. Finally, we conclude with a series of examples that demonstrate the strength and versatility of our new approach.
\section{\textbf{Preliminaries}}\label{s:prelim}

\subsection{\textbf{Main problems}}
The concurrence is an entanglement monotone of a mixed bipartite quantum state. For a pure
state $\rho= |\psi\rangle_{A B}\in\mathcal{H}_A \otimes \mathcal{H}_B$, it is defined by \cite{AF,RBCGM,U}
\begin{equation}\label{e:Pure}
C\left(|\psi\rangle_{A B}\right)=\sqrt{2\left[1-\operatorname{Tr}\left(\rho_A^2\right)\right]}=\sqrt{2\left[1-\operatorname{Tr}\left(\rho_B^2\right)\right]},
\end{equation}
where $\rho_A$ (resp. $\rho_B$) is the reduced density matrix by tracing over the subsystem $B$ (resp. $A$).

For a mixed state $\rho_{AB}$, 
the concurrence and the concurrence of assistance (CoA) are defined respectively by \cite{DFMSTU, YS}
 \begin{equation}\label{e:Mixed}
 C\left(\rho_{A B}\right)=\min _{\left\{p_i,\left|\psi_i\right\rangle\right\}} \sum_i p_i C\left(\left|\psi_i\right\rangle\right), \
C_{a}\left(\rho_{A B}\right)=\max _{\left\{p_i,\left|\psi_i\right\rangle\right\}} \sum_i p_i C\left(\left|\psi_i\right\rangle\right),
\end{equation}
where the minimum/maximum is taken over all possible convex roofs of pure state decompositions: $\rho_{A B}=\sum_i p_i\left|\psi_i\right\rangle\left\langle\psi_i\right|$ with $p_i \geqslant 0, \sum_i p_i=1$ and $\left|\psi_i\right\rangle \in \mathcal{H}_A \otimes \mathcal{H}_B$.

\medskip
Recall that the linear entropy of a state $\rho$ is defined as \cite{SF}:
\begin{equation}\label{e:entropy}
T(\rho)=1-\operatorname{Tr}\left(\rho^2\right)
\end{equation}
and for a bipartite state $\rho_{A B}$,
$T(\rho_{A B})$ has the following property \cite{ZGZG}:
\begin{equation}\label{e:property}
|T(\rho_{A})-T(\rho_{B})|\leq T(\rho_{AB})\leq T(\rho_{A})+T(\rho_{B}).
\end{equation}

\medskip

For any $N$-qubit pure state $|\psi\rangle_{ABC_{1}\cdots C_{N-2}}$, it follows from \eqref{e:Pure} and \eqref{e:entropy} that
\begin{equation}\label{e:Pure-N}
 C^{2}(|\psi\rangle_{AB|C_{1}\cdots B_{N-2}})=2\left[1-\operatorname{Tr}\left(\rho_{AB}^2\right)\right]=2T(\rho_{AB}).
\end{equation}

 \bigskip

The authors \cite{YCLZ,JF,JFL,SXZWF} have presented some generalized monogamy and generalized polygamy relations of concurrence for $N$-qubit systems under the partition $AB|C_{1}\cdots C_{N-2}$, where lower (resp. upper) bound of the monogamy (resp. polygamy) relation for the $\alpha$th $(0\leq \alpha \leq 2)$ power of concurrence were considered.
We will focus on two questions in this paper:
\begin{enumerate}[label= {\bf(\roman*)}]
\item Tighter generalized monogamy relation: $C^{\alpha}(|\psi\rangle_{AB|C_{1}\cdots B_{N-2}})=(2T(\rho_{AB}))^{\frac{\alpha}{2}}\geq |2T(\rho_{A})-2T(\rho_{B})|^{\frac{\alpha}{2}}
    =| C^{2}(|\psi\rangle_{A|BC_{1}\cdots C_{N-2}})- C^{2}(|\psi\rangle_{B|AC_{1}\cdots C_{N-2}})|^{\frac{\alpha}{2}}\geq \mathrm{optimal\ lower \ bound}?$
\item Tighter generalized polygamy relation: $C^{\alpha}(|\psi\rangle_{AB|C_{1}\cdots C_{N-2}})=(2T(\rho_{AB}))^{\frac{\alpha}{2}}\leq (2T(\rho_{A})+2T(\rho_{B}))^{\frac{\alpha}{2}}
    =| C^{2}(|\psi\rangle_{A|BC_{1}\cdots C_{N-2}})+ C^{2}(|\psi\rangle_{B|AC_{1}\cdots C_{N-2}})|^{\frac{\alpha}{2}}\leq \mathrm{optimal \ upper\  bound}?$
\end{enumerate}

\bigskip
 To approach these problems, 
 we consider how to bound the binomial function $(1+t)^x$ over an interval. There are several useful inequalities in the literature corresponding to various monogamy and polygamy relations, our idea is to formulate some
 parametrized strong inequality, which behaves like a master
 inequality that includes many previously known ones as special examples and provides optimal monogamy and polygamy relations than some of the recently known ones.

\subsection{\textbf{A key inequality}}
We start by giving the following inequality.
\begin{lem}\label{l:chap2-lem1}
For a given $p\in (0, 1]$, let $0\leq t \leq p$ and $0\leq x \leq 1$, then
\begin{equation}\label{e:chap2-ineq1}
\begin{aligned}
(1+t)^{x}\leq 1+\frac{(1+p)^{x}-1}{p^{x}}t^{x}+\frac{xp}{(1+p)^{2}}\left(\frac{t}{p}-\left(\frac{t}{p}\right)^{x}\right)=\Omega+\Upsilon t^{x},
\end{aligned}
\end{equation}
where $\Omega=1+\frac{xt}{(1+p)^{2}}$ and $\Upsilon=\frac{(1+p)^{x}-1}{p^{x}}-\frac{xp}{(1+p)^{2}p^{x}}$.
\end{lem}

\bigskip
\begin{proof}
The inequality holds when $t=0 ~ (x\neq 0)$. 
Now for fixed positive $p \ (\leq1)$, consider the following function:
$$f(x,y)=\frac{(1+y)^{x}-\frac{xy}{(1+p)^{2}}-1}{y^{x}}$$
defined on $[0,1]\times(0,p].$ Then $\frac{\partial f(x,y)}{\partial y}=\frac{xy^{x-1}\left[-(1+y)^{x-1}+\frac{(x-1)y}{(1+p)^{2}}+1\right]}{y^{2x}}.$ 

Let $h(x,y)=(1+y)^{x-1}-\frac{(x-1)y}{(1+p)^{2}}-1$, $(0\leq y \leq p, 0\leq x\leq 1)$, then 
$\frac{\partial h(x,y)}{\partial y}=(x-1)\left[(1+y)^{x-2}-\frac{1}{(1+p)^{2}}\right]\leq 0$, This means that $h(x,y)$ is decreasing as a function of $y$. Subsequently, we have $h(x,y)\leq h(x,0)=0.$ Therefore $\frac{\partial f(x,y)}{\partial y}\geq 0$, which immeditately implies that $f(x,y)$ is increasing as a function of $y$. Thus, for $0<t\leq p$, we have $f(x,t)\leq f(x,p)$, which is \eqref{e:chap2-ineq1}.
\end{proof}

\bigskip

In fact, the relation given in \cite[Eq. (12)]{SXZWF} is just a special case of our  \eqref{e:chap2-ineq1} for $p=1$.
 Note that $\frac{t}{p}-\left(\frac{t}{p}\right)^{x}\leq 0$ when $0<t\leq p\leq 1$ and $0\leq x \leq 1$, therefore
 \begin{equation}\label{e:chap2-ineq2}
\begin{aligned}
(1+t)^{x} &\leq 1+\frac{(1+p)^{x}-1}{p^{x}}t^{x}+\frac{xp}{(1+p)^{2}}\left(\frac{t}{p}-\left(\frac{t}{p}\right)^{x}\right)\\
&\leq 1+\frac{(1+p)^{x}-1}{p^{x}}t^{x}\leq 1+(2^{x}-1)t^{x}\\
&\leq 1+xt^{x}\leq 1+t^{x}.
\end{aligned}
\end{equation}
Note that the three inequalities above correspond exactly to those considered in \cite{SXZWF, YCLZ,JF,JFL}. 
Therefore any polygamy relations based on Lemma \ref{l:chap2-lem1} will produce tighter bounds than those given in \cite{SXZWF} as well as \cite{YCLZ,JF,JFL}.

\bigskip
\section{\textbf{Tighter $\beta$th power of assisted entanglement measures}} \label{s:poly}
Let $\rho=\rho_{AB_{1}\cdots B_{N-1}}$ be a quantum state over the Hilbert space $\mathcal{H}_{A}\bigotimes $ $ \mathcal{H}_{B_{1}} $ $
\bigotimes $ $ \cdots$ $ \bigotimes \mathcal{H}_{B_{N-1}}$. From now on, if the state $\rho$ is clear from the context, we will abbreviate the concurrence and concurrence assisted by $C(\rho_{AB_{i}})=C_{AB_{i}}$, $C_{a}(\rho_{AB_{i}})=C_{aAB_{i}}$ (cf. \eqref{e:Mixed}).

\subsection{\textbf{Parameterized polygamy relation}}\label{sub:chap3.1}
Consider an $N$-partite pure state $|{\psi\rangle}_{AB_1\cdots B_{N-1}}$, the concurrence with respect to the partition $A|B_1\cdots B_{N-1}$ obeys
the polygamy relation \cite{GBS}:
\begin{equation}\label{e:mono-N}
C^{2}(|\psi\rangle_{A|B_{1}\cdots B_{N-1}})\leq \sum_{i=1}^{N} C^{2}_{a}(\rho_{AB_{i}}).
\end{equation}
One can group the summands into $k$ terms by
partitioning the indices $\{1, 2, \ldots, N-1\}$ into $k$ subsets.
Namely we rewrite \eqref{e:mono-N} as
\begin{equation}\label{e:chap3.1-ineq1}
C^{2}(|\psi\rangle_{A|B_{1}\cdots B_{N-1}})\leq \sum_{i=1}^{k} C^{2}_{a}(\rho_{AN_{i}}).
\end{equation}
where $C^{2}_{a}(\rho_{AN_{i}})=\sum_{j=N_{i-1}+1}^{N_{i-1}+N_{i}} C^{2}_{a}(\rho_{AB_{j}})>0$ with $N_{0}=0$ and $\sum_{i=1}^{k}N_{i}=N-1$, $1\leq k \leq N-1.$

\bigskip
Using Lemma \ref{l:chap2-lem1}, we obtain the following results.

\begin{thm}\label{t:chap3.1-thm1} Let $|\psi\rangle_{AB_{1}\cdots B_{N-1}}$ be any $N$-qubit pure state and $C_{a}$ be bipartite assisted quantum measure CoA satisfying the polygamy relation \eqref{e:chap3.1-ineq1}. If there exist $0<p_{i}\leq 1$ such that $p_{i}C^{2}_{a}(\rho_{AN_{i}})\geq \sum_{j=i+1}^{k} C^{2}_{a}(\rho_{AN_{j}})$ for each $i=1,2,\cdots,k-1$, then
\begin{equation}\label{e:chap3.1-ineq2}
C^{\beta}(|\psi\rangle_{A|B_{1}\cdots B_{N-1}})\leq \sum_{l=1}^{k-1}\prod_{u=0}^{l-1}\Upsilon_{u}\Omega_{l}C_{a}^{\beta}(\rho_{AN_{l}})+\prod_{u=1}^{k-1}\Upsilon_{u}C_{a}^{\beta}(\rho_{AN_{k}}),
\end{equation}
for all $0\leq \beta \leq 2$ and $N\geq 3$, where $\Omega_{l}=1+\frac{\beta}{2(1+p_{l})^{2}}\frac{\sum_{j=l+1}^{k}C_{a}^{2}(\rho_{AN_{j}})}{C_{a}^{2}(\rho_{AN_{l}})}, l=1,2,\cdots, k-1;$
$\Upsilon_{u}=\frac{(1+p_{u})^{\frac{\beta}{2}}-1}{p_{u}^{\frac{\beta}{2}}}-\frac{\beta p_{u}}{2p_{u}^{\frac{\beta}{2}}(1+p_{u})^{2}},u=1,2,\cdots, k-1$ and $\Upsilon_{0}=1$.
\end{thm}
\begin{proof} For each $u=1, \ldots, k-2$ it follows from Lemma \ref{l:chap2-lem1} that
\begin{align} \notag
&\left(\sum_{i=u}^{k} C^{2}_{a}(\rho_{AN_{i}})\right)^{\frac{\beta}{2}}=C^{\beta}_{a}(\rho_{AN_{u}})\left(1+\frac{\sum_{i=u+1}^{k} C^{2}_{a}(\rho_{AN_{i}})}{C^{\beta}_{a}(\rho_{AN_{u}})}\right)^{\frac{\beta}{2}}\\ \label{e:iter}
&\leq \Omega_{u}C^{\beta}_{a}(\rho_{AN_{u}})+\Upsilon_u\left(\sum_{i=u+1}^{k} C^{2}_{a}(\rho_{AN_{i}})\right)^{\frac{\beta}{2}}
\end{align}

Using \eqref{e:chap3.1-ineq1} one has that
\begin{align*}
&C^{\beta}(|\psi\rangle_{A|B_{1}\cdots B_{N-1}})\leq 
\left(\sum_{i=1}^{k} C^{2}_{a}(\rho_{AN_{i}})\right)^{\frac{\beta}{2}}\\
&\leq\Upsilon_{1}\left(\sum_{i=2}^{k} C^{2}_{a}(\rho_{AN_{i}})\right)^{\frac{\beta}{2}}+\Omega_{1}C^{\beta}_{a}(\rho_{AN_{1}})~~~(\text{by \eqref{e:iter}})\\
&\leq\Upsilon_{1}\left(\Upsilon_{2}\left(\sum_{i=3}^{k} C^{2}_{a}(\rho_{AN_{i}})\right)^{\frac{\beta}{2}}+\Omega_{2}C^{\beta}_{a}(\rho_{AN_{2}})\right)+\Omega_{1}C^{\beta}_{a}(\rho_{AN_{1}})\\
&\leq\cdots\leq\prod_{u=1}^{k-1}\Upsilon_{u}C_{a}^{\beta}(\rho_{AN_{k}})+\sum_{l=1}^{k-1}\prod_{u=0}^{l-1}\Upsilon_{u}\Omega_{l}C_{a}^{\beta}(\rho_{AN_{l}}) \quad \quad (\text{iterating \eqref{e:iter}}).
\end{align*}
\end{proof}
In particular, if $k=N-1,$ then we have

\begin{equation}\label{e:chap3.1-ineq3}
C^{\beta}(|\psi\rangle_{A|B_{1}\cdots B_{N-1}})\leq \sum_{l=1}^{N-2}\prod_{u=0}^{l-1}\Upsilon_{u}^{'}\Omega_{l}^{'}C_{a}^{\beta}(\rho_{AB_{l}})+\prod_{u=1}^{N-2}\Upsilon_{u}^{'}C_{a}^{\beta}(\rho_{AB_{N-1}}),
\end{equation}
for all $0\leq \beta \leq 2$ and $N\geq 3$, where $\Omega_{l}^{'}=1+\frac{\beta}{2(1+p_{l})^{2}}\frac{\sum_{j=l+1}^{N-1}C_{a}^{2}(\rho_{AB_{j}})}{C_{a}^{2}(\rho_{AB_{l}})}, l=1,2,\cdots, N-2;$
$\Upsilon_{u}^{'}=\frac{(1+p_{u})^{\frac{\beta}{2}}-1}{p_{u}^{\frac{\beta}{2}}}-\frac{\beta p_{u}}{2p_{u}^{\frac{\beta}{2}}(1+p_{u})^{2}},u=1,2,\cdots, N-2$ and $\Upsilon_{0}^{'}=1$.

\subsection{\textbf{Comparison of parameterized polygamy relations of CoA}}
 Based on Theorem \ref{t:chap3.1-thm1} and \eqref{e:chap2-ineq2}, we have the following polygamy relations of $\beta$th $(0\leq \beta \leq 2)$ power of CoA.
\begin{equation}\label{e:chap3.1-ineq4}
\begin{aligned}
C^{\beta}(|\psi\rangle_{A|B_{1}\cdots B_{N-1}})&\leq \sum_{l=1}^{k-1}\prod_{u=0}^{l-1}\Upsilon_{u}\Omega_{l}C_{a}^{\beta}(\rho_{AN_{l}})+\prod_{u=1}^{k-1}\Upsilon_{u}C_{a}^{\beta}(\rho_{AN_{k}}) \quad (\text{by  \ref{e:chap3.1-ineq2}})\\
&\leq \sum_{l=1}^{k}\prod_{u=0}^{l-1}\Phi_{u}C_{a}^{\beta}(\rho_{AN_{l}}) ~~~~~ (\text{ by $(1+t)^{x}\leq1+\frac{(1+p)^{x}-1}{p^{x}}t^{x}$ in  \eqref{e:chap2-ineq2}})\\
&\leq \sum_{l=1}^{k}(2^{\frac{\beta}{2}}-1)^{l-1}C_{a}^{\beta}(\rho_{AN_{l}})  \quad \quad(\text{the inequality in \cite{JFL}})\\
&\leq \sum_{l=1}^{k}(\frac{\beta}{2})^{l-1}C_{a}^{\beta}(\rho_{AN_{l}}) \quad \quad (\text{the inequality in \cite{JF}})\\
&\leq \sum_{l=1}^{k}C_{a}^{\beta}(\rho_{AN_{l}})  \quad \quad(\text{by $(1+t)^{x}\leq 1+t^{x}$ in  \eqref{e:chap2-ineq2}})
\end{aligned}
\end{equation}
where $\Phi_{u}=((1+p_{i})^{\frac{\beta}{2}}-1){p_{i}^{-\frac{\beta}{2}}}$ and $0<p_{i}\leq 1$, $i=1,2,\cdots,k-1$, $1\leq k \leq N-1.$

\bigskip
 In fact, the polygamy relation given in \cite[Thm. 2]{SXZWF} is a special case of our Theorem \ref{t:chap3.1-thm1}
 with all $p_{i}=1$, so our bound is tighter due to \eqref{e:chap3.1-ineq4}. Since the bound of \cite{SXZWF} is tighter than those of \cite{YCLZ,JF,JFL}, our
 parameterized $(0<p_{i}\leq 1$, $i=1,2,\cdots,k-1)$ polygamy relations based on Lemma \ref{l:chap2-lem1} is
 tighter than those given in \cite{YCLZ,JF,JFL} as well. 

\medskip
Now we use a specific example to show how our bounds for polygamy relations outperform these other studies (see Fig 1).
\begin{exmp}
Let $\rho=|\psi\rangle_{ABC}\langle\psi|$ be the three-qubit state \cite{AACJLT}:
$$
|\psi\rangle_{ABC}=\lambda_{0}|000\rangle+\lambda_{1}e^{i \varphi}|100\rangle+\lambda_{2}|101\rangle+\lambda_{3}|110\rangle+\lambda_{4}|111\rangle.
$$
where $\sum_{i=0}^4 \lambda_i^2=1$, and $\lambda_{i}\geq 0$ for $i=0,1,2,3,4$.
Then  $C(|\psi\rangle_{A|BC})=2 \lambda_0 \sqrt{\lambda_2^2+\lambda_3^2+\lambda_4^2}$, $C_{aAB}=2 \lambda_0 \sqrt{\lambda^{2}_{2}+\lambda^{2}_{4}}$, and $C_{aAC}=2 \lambda_0 \sqrt{\lambda^{2}_{3}+\lambda^{2}_{4}}$.

Set $\lambda_0=\lambda_3=\frac{1}{2}, \lambda_1=\lambda_4=0, \lambda_2=\frac{\sqrt{2}}{2}$.
Then we have $C(|\psi\rangle_{A|BC})=\frac{\sqrt{3}}{2}, C_{aAB }=\frac{\sqrt{2}}{2}, C_{aAC}=\frac{1}{2}$. Thus $t=\frac{C_{aAC}^{2}}{C_{aAB}^{2}}=\frac{1}{2}$. Set $p = \frac{3}{5}$ (since $0<t\leq p \leq 1$).

By Theorem \ref{t:chap3.1-thm1}, for any $0\leq\beta\leq 2$, the RHS of the polygamy relation is
\begin{align*}
Z_1&=\left(1+\frac{\beta}{2(1+p)^{2}}\frac{C_{aAC}^{2}}{C_{aAB}^{2}}\right)C_{aAB}^{\beta}+\left[\frac{(1+p)^{\frac{\beta}{2}}-1}{p^{\frac{\beta}{2}}}-\frac{\beta p}{2p^{\frac{\beta}{2}}(1+p)^{2}}\right]C_{aAC}^{\beta}\\
&=\left(1+\frac{25}{256}\beta\right)\left(\frac{\sqrt{2}}{2}\right)^{\beta}+\left[\frac{\left(\frac{8}{5}\right)^{\frac{\beta}{2}}-1-\frac{15}{128}\beta}{\left(\frac{3}{5}\right)^{\frac{\beta}{2}}}\right]\left(\frac{1}{2}\right)^{\beta}.
\end{align*}

The RHS of the polygamy relation from \cite{SXZWF} is a special case of our bound at $p=1$:
\begin{align*}
    Z_2&=\left(1+\frac{\beta}{8}\frac{C_{aAC}^{2}}{C_{aAB}^{2}}\right)C_{aAB}^{\beta}+\left[2^{\frac{\beta}{2}}-1-\frac{\beta }{8}\right]C_{aAC}^{\beta}\\
&=\left(1+\frac{\beta}{16}\right)\left(\frac{\sqrt{2}}{2}\right)^{\beta}+\left[2^{\frac{\beta}{2}}-1-\frac{\beta }{8}\right]\left(\frac{1}{2}\right)^{\beta}.
\end{align*}

The RHS of the polygamy relation given in \cite{JF,JFL}
are respectively:
\begin{align*}
    Z_{_3}=\left(\frac{\sqrt{2}}{2}\right)^{\beta}+\left(2^{\frac{\beta}{2}}-1\right)\left(\frac{1}{2}\right)^{\beta},~~
    Z_{_4}= \left(\frac{\sqrt{2}}{2}\right)^{\beta}+\frac{\beta}{2}\left(\frac{1}{2}\right)^{\beta}.
\end{align*}

\begin{figure}[H] 
    \centering 
    \includegraphics[width=0.8\textwidth]{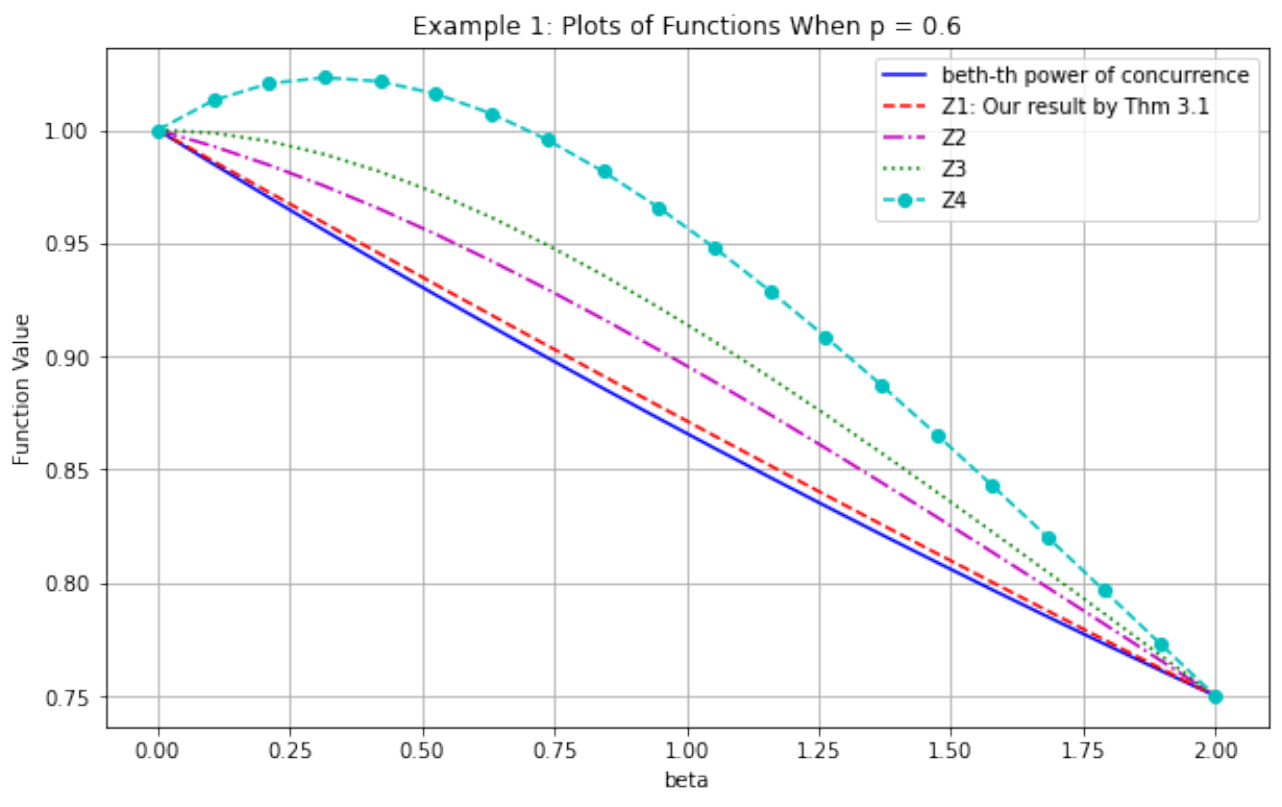} 
\end{figure}

Figure 1: The $x$-axis is $\beta$. From bottom to top, the solid blue line is the $\beta$th power of concurrence for $C(|\psi\rangle_{A|BC})$, the dashed red line $Z_1$ represents the upper bound from our result $(3.2)$ in Theorem \ref{t:chap3.1-thm1}, the dash-dot magenta line $Z_2$ represents the result in \cite{SXZWF}, the dotted green line $Z_3$ represents the result in \cite{JF}, and the dashed cyan line with circle markers $Z_4$ represents the result in \cite{JFL}. The graph shows our upper bound $Z_1$ is the strongest.
\end{exmp}

\section{\textbf{Optimized bounds}} \label{s:optimized bound}

Combining with Section \ref{s:poly}, we can obtain the optimized upper bounds of the polygamy relation for the $\beta$th $(0\leq \beta \leq 2)$ power of CoA. Based on this, we will then follow the conclusions of Section \ref{s:poly} to derive the optimized lower (resp. upper) bounds of the generalized monogamy (resp. polygamy) relation for the $\alpha$th $(0\leq \alpha \leq 2)$ power of concurrence about $N$-qubit pure state $|\psi\rangle_{ABC_{1}\cdots C_{N-2}}$ under the partition $AB$ and $C_{1}\cdots C_{N-2}$.

\subsection{\textbf{The optimized bounds of $C^{\alpha}(|\psi\rangle_{AB|C_{1}\cdots C_{N-2}})$ }} \label{sub:optimized bound-C}
It follows from \eqref{e:property}-\eqref{e:Pure-N}
that
 \begin{enumerate}[label= {\bf(\roman*)}]
\item The generalized monogamy: $C^{\alpha}(|\psi\rangle_{AB|C_{1}\cdots B_{N-2}})=(2T(\rho_{AB}))^{\frac{\alpha}{2}}\geq |2T(\rho_{A})-2T(\rho_{B})|^{\frac{\alpha}{2}}\\
    =| C^{2}(|\psi\rangle_{A|BC_{1}\cdots C_{N-2}})- C^{2}(|\psi\rangle_{B|AC_{1}\cdots C_{N-2}})|^{\frac{\alpha}{2}}$
\item The generalized polygamy: $C^{\alpha}(|\psi\rangle_{AB|C_{1}\cdots C_{N-2}})=(2T(\rho_{AB}))^{\frac{\alpha}{2}}\leq (2T(\rho_{A})+2T(\rho_{B}))^{\frac{\alpha}{2}}\\
    =| C^{2}(|\psi\rangle_{A|BC_{1}\cdots C_{N-2}})+ C^{2}(|\psi\rangle_{B|AC_{1}\cdots C_{N-2}})|^{\frac{\alpha}{2}}$
\end{enumerate}

\medskip
The following results are easy consequences of \eqref{e:chap3.1-ineq1} and \eqref{e:chap3.1-ineq2} with the help of the key lemma.
\begin{prop}\label{p:bounds}
For $0\leq \alpha \leq 2$ and any $N$-qubit state $|\psi\rangle_{AB|C_{1}\cdots C_{N-2}}\, (N\geq4)$, 
\begin{enumerate}[label={\bf Case \arabic*.}]
\item The lower bound of generalized monogamy for $C^{\alpha}(|\psi\rangle_{AB|C_{1}\cdots C_{N-2}})$:
\begin{equation}\label{e:chap4.1-ineq1}
\begin{aligned}
 C^{\alpha}(|\psi\rangle_{AB|C_{1}\cdots B_{N-2}})\geq \max \left\{\left(\sum_{i=1}^{N-2}C^{2}_{AC_{i}}+C^{2}_{AB}\right)^{\frac{\alpha}{2}}-\Theta_{B},\left(\sum_{i=1}^{N-2}C^{2}_{BC_{i}}+C^{2}_{AB}\right)^{\frac{\alpha}{2}}-\Theta_{A}\right\}
 \end{aligned}
 \end{equation}
\item The upper bound of generalized polygamy for $C^{\alpha}(|\psi\rangle_{AB|C_{1}\cdots C_{N-2}})$:
\begin{equation}\label{e:chap4.1-ineq2}
\begin{aligned}
C^{\alpha}(|\psi\rangle_{AB|C_{1}\cdots C_{N-2}})\leq \Theta_{A}+\Theta_{B}
   \end{aligned}
\end{equation}
 \end{enumerate}
 where  $$\Theta_{j}=\sum_{l=1}^{k_{j}-1}\prod_{u=0}^{l-1}\Upsilon_{u_{j}}\Omega_{l_{j}}C_{a}^{\alpha}(\rho_{jN_{l_{j}}})+\prod_{u=1}^{k_{j}-1}\Upsilon_{u_{j}}C_{a}^{\alpha}(\rho_{jN_{k_{j}}}), $$
with $\Omega_{l_{j}}=1+\frac{\alpha}{2(1+p_{l_{j}})^{2}}\frac{\sum_{i=l+1}^{k_{j}}C_{a}^{2}(\rho_{jN_{i_{j}}})}{C_{a}^{2}(\rho_{jN_{l_{j}}})}, l=1,2,\cdots, k_{j}-1;$ $\Upsilon_{0_{j}}=1$ and $\Upsilon_{u_{j}}=\frac{(1+p_{u_{j}})^{\frac{\alpha}{2}}-1}{p_{u_{j}}^{\frac{\alpha}{2}}}-\frac{\alpha p_{u_{j}}}{2p_{u_{j}}^{\frac{\alpha}{2}}(1+p_{u_{j}})^{2}}$, $u=1,2,\cdots, k_{j}-1$, $2\leq k_{j} \leq N-1,$  $j=A,B.$
\end{prop}
\begin{proof}
If $C^{2}(|\psi\rangle_{A|BC_{1}\cdots C_{N-2}}) \leq C^{2}(|\psi\rangle_{B|AC_{1}\cdots C_{N-2}})$, then
\begin{equation}\label{e:chap4.1-ineq3}
\begin{aligned}
& C^{\alpha}(|\psi\rangle_{AB|C_{1}\cdots B_{N-2}})\geq| C^{2}(|\psi\rangle_{A|BC_{1}\cdots C_{N-2}})- C^{2}(|\psi\rangle_{B|AC_{1}\cdots C_{N-2}})|^{\frac{\alpha}{2}} \\
     &\geq C^{\alpha}(|\psi\rangle_{B|AC_{1}\cdots C_{N-2}})- C^{\alpha}(|\psi\rangle_{A|BC_{1}\cdots C_{N-2}})\quad \quad(\text{by \cite[Lemma]{JF}})\\
     &\geq \left(\sum_{i=1}^{N-2}C^{2}_{BC_{i}}+C^{2}_{AB}\right)^{\frac{\alpha}{2}}- C^{\alpha}(|\psi\rangle_{A|BC_{1}\cdots C_{N-2}})\quad \quad(
     \text{by \eqref{e:mono-N}})
   \end{aligned}
 \end{equation}
For simplicity, denote  $|\psi\rangle_{A|D_{1}D_{2}\cdots D_{N-1}}=|\psi\rangle_{A|BC_{1}\cdots C_{N-2}},$ where $D_{1}=B,D_{i}=C_{i-1},i=2,\cdots,N-1.$ Thus, it can be directly deduced from  \eqref{e:chap3.1-ineq1} and Theorem \ref{t:chap3.1-thm1} that there exists $0<p_{i_{A}}\leq 1$ such that $p_{i_{A}}C^{2}_{a}(\rho_{AN_{i_{A}}})\geq \sum_{j_{A}=i_{A}+1}^{k_{A}} C^{2}_{a}(\rho_{AN_{j_{A}}})$ for $i=1,2,\cdots,k_{A}-1$, then
\begin{equation}\label{e:chap4.1-ineq4}
C^{\alpha}(|\psi\rangle_{A|BC_{1}\cdots C_{N-2}})\leq \sum_{l=1}^{k_{A}-1}\prod_{u=0}^{l-1}\Upsilon_{u_{A}}\Omega_{l_{A}}C_{a}^{\alpha}(\rho_{AN_{l_{A}}})+\prod_{u=1}^{k_{A}-1}\Upsilon_{u_{A}}C_{a}^{\alpha}(\rho_{AN_{k_{A}}}),
\end{equation}
for all $0\leq \alpha \leq 2$, where $\Omega_{l_{A}}=1+\frac{\alpha}{2(1+p_{l_{A}})^{2}}\frac{\sum_{i=l+1}^{k_{A}}C_{a}^{2}(\rho_{AN_{i_{A}}})}{C_{a}^{2}(\rho_{AN_{l_{A}}})}, l=1,2,\cdots, k_{A}-1;$ $\Upsilon_{0_{A}}=1$ and $\Upsilon_{u_{A}}=\frac{(1+p_{u_{A}})^{\frac{\alpha}{2}}-1}{p_{u_{A}}^{\frac{\alpha}{2}}}-\frac{\alpha p_{u_{A}}}{2p_{u_{A}}^{\frac{\alpha}{2}}(1+p_{u_{A}})^{2}},u=1,2,\cdots, k_{A}-1$, $2\leq k_{A} \leq N-1.$

 \bigskip
Similarly, if $C^{2}(|\psi\rangle_{B|AC_{1}\cdots C_{N-2}}) \leq C^{2}(|\psi\rangle_{A|BC_{1}\cdots C_{N-2}})$, we then have
\begin{equation}\label{e:chap4.1-ineq5}
C^{\alpha}(|\psi\rangle_{B|AC_{1}\cdots C_{N-2}})\leq \sum_{l=1}^{k_{B}-1}\prod_{u=0}^{l-1}\Upsilon_{u_{B}}\Omega_{l_{B}}C_{a}^{\alpha}(\rho_{BN_{l_{B}}})+\prod_{u=1}^{k_{B}-1}\Upsilon_{u_{B}}C_{a}^{\alpha}(\rho_{BN_{k_{B}}}),
\end{equation}
for all $0\leq \alpha \leq 2$, where $\Omega_{l_{B}}=1+\frac{\alpha}{2(1+p_{l_{B}})^{2}}\frac{\sum_{i=l+1}^{k_{B}}C_{a}^{2}(\rho_{BN_{i_{B}}})}{C_{a}^{2}(\rho_{BN_{l_{B}}})}, l=1,2,\cdots, k_{B}-1;$ $\Upsilon_{0_{B}}=1$ and $\Upsilon_{u_{B}}=\frac{(1+p_{u_{B}})^{\frac{\alpha}{2}}-1}{p_{u_{B}}^{\frac{\alpha}{2}}}-\frac{\alpha p_{u_{B}}}{2p_{u_{B}}^{\frac{\alpha}{2}}(1+p_{u_{B}})^{2}},u=1,2,\cdots, k_{B}-1$, $2\leq k_{B} \leq N-1.$

Plugging \eqref{e:chap4.1-ineq4} or \eqref{e:chap4.1-ineq5} into \eqref{e:chap4.1-ineq3} respectively, we obtain \eqref{e:chap4.1-ineq1}.
Meanwhile,
\begin{equation}\label{e:chap4.1-ineq6}
\begin{aligned}
&C^{\alpha}(|\psi\rangle_{AB|C_{1}\cdots C_{N-2}}) 
\leq | C^{2}(|\psi\rangle_{A|BC_{1}\cdots C_{N-2}})+ C^{2}(|\psi\rangle_{B|AC_{1}\cdots C_{N-2}})|^{\frac{\alpha}{2}}\\
   &\leq C^{\alpha}(|\psi\rangle_{A|BC_{1}\cdots C_{N-2}})+ C^{\alpha}(|\psi\rangle_{B|AC_{1}\cdots C_{N-2}})\quad \quad(\text{ by Lemma in \cite{JF}})\\
  & \leq \Theta_{A}+\Theta_{B}\quad \quad\quad\quad\quad(\text{by  \eqref{e:chap4.1-ineq4} and  \eqref{e:chap4.1-ineq5}})
   \end{aligned}
\end{equation}

\end{proof}

Based on the discussion, the following 
are now clear.
\begin{enumerate}[label= {\bf(\roman*)}]
\item The lower bound of generalized monogamy for $C^{\alpha}(|\psi\rangle_{AB|C_{1}\cdots C_{N-2}})$ given in \cite[Thm. 3]{SXZWF} is a special case of our  \eqref{e:chap4.1-ineq1} in which all $p_{i_{j}}=1,i=1,2,\cdots,k_{j}-1,j=A,B.$ Combining with  \eqref{e:chap3.1-ineq4}, we claim that  \eqref{e:chap4.1-ineq1} is the optimized lower bound of generalized monogamy relation than \cite[Thm. 3]{JF}, \cite[Thm. 2]{JFL}.
\item The upper bound of generalized polygamy for $C^{\alpha}(|\psi\rangle_{AB|C_{1}\cdots C_{N-2}})$ given in \cite[Thm. 4]{SXZWF} is a special case of our  \eqref{e:chap4.1-ineq2} in which all $p_{i_{j}}=1,i=1,2,\cdots,k_{j}-1,j=A,B.$ Combining with  \eqref{e:chap3.1-ineq4}, we claim that  \eqref{e:chap4.1-ineq2} is the optimized upper bound of generalized polygamy relation than \cite[Thm. 4]{JF}, \cite[Thm. 3]{JFL}.
\end{enumerate}

\bigskip
 We provide an example to illustrate the comparison (see Figs 2 and 3).
\begin{exmp}
Consider the following $4$-qubit generalized $W$-class state \cite{YCLZ}:
$$
|W\rangle_{ABC_{1}C_{2}}=\lambda_{1}(|1000\rangle+\lambda_{2}|0100\rangle)+\lambda_{3}|0010\rangle+\lambda_{4}|0001\rangle.
$$
 where $\sum_{i=1}^4 \lambda_i^2=1$, and $\lambda_{i}\geq 0$ for $i=1,2,3,4$.
Then  $C(|W\rangle_{AB|C_{1}C_{2}})=2  \sqrt{(\lambda_1^2+\lambda_2^2)(\lambda_3^2+\lambda_4^2)}$, $C_{AB}=C_{aAB}=2 \lambda_1\lambda_2 $, $C_{AC_{1}}=C_{aAC_{1}}=2 \lambda_1\lambda_3 $, $C_{AC_{2}}=C_{aAC_{2}}=2 \lambda_1\lambda_4 $, $C_{BC_{1}}=C_{aBC_{1}}=2 \lambda_2\lambda_3 $, and $C_{BC_{2}}=C_{aBC_{2}}=2 \lambda_2\lambda_4 $.

\bigskip
\begin{enumerate}[label={\bf Case \arabic*.}]
\item (Monogamy, lower bounds) 
Set $\lambda_1=\frac{3}{4}, \lambda_2=\frac{1}{2}, \lambda_3=\frac{\sqrt{2}}{4}, \lambda_4=\frac{1}{4}$ (which satisfy the hypothesis of Thm. \ref{t:chap3.1-thm1}). By \eqref{e:chap4.1-ineq1}, our lower bound is
$
T_{1}=\left(\sum_{i=1}^{2}C^{2}_{AC_{i}}+C^{2}_{AB}\right)^{\frac{\alpha}{2}}-\Theta_{B}.
 $
 Using Theorem \ref{t:chap3.1-thm1} and letting $k_{B}=3$, then $t_{1_{B}}=\frac{C^{2}_{aBC_{1}}+C^{2}_{aBC_{2}}}{C^{2}_{aAB}}=\frac{1}{3}$, $t_{2_{B}}=\frac{C^{2}_{aBC_{2}}}{C^{2}_{aBC_{1}}}=\frac{1}{2}$. 
Thus by Lemma \ref{l:chap2-lem1}, set $p_{1_{B}}=\frac{2}{5}(>\frac{1}{3})$, $p_{2_{B}}=\frac{3}{5}(>\frac{1}{2})$. 
With these data, our lower bound is
\begin{align*}
T_{1}=\left(C^{2}_{AB}+C^{2}_{AC_{1}}+C^{2}_{AC_{2}}\right)^{\frac{\alpha}{2}}-\Theta_{B}=\left(\frac{63}{64}\right)^{\frac{\alpha}{2}}-\Theta_{B},
 \end{align*}
where by Prop. \ref{p:bounds}
\begin{equation}\label{e:chap4.1-ineq7}
\begin{aligned}
\Theta_{B}
&=\left(1+\frac{25}{294}\alpha\right)\left(\frac{3}{4}\right)^{\alpha}+\left(\frac{\left(\frac{7}{5}\right)^{\frac{\alpha}{2}}-1-\frac{5}{49}\alpha}{\left(\frac{2}{5}\right)^{\frac{\alpha}{2}}}\right)\left(1+\frac{25}{256}\alpha\right)\left(\frac{\sqrt{2}}{4}\right)^{\alpha}\\
&+\left(\frac{\left(\frac{7}{5}\right)^{\frac{\alpha}{2}}-1-\frac{5}{49}\alpha}{\left(\frac{2}{5}\right)^{\frac{\alpha}{2}}}\right)\left(\frac{\left(\frac{8}{5}\right)^{\frac{\alpha}{2}}-1-\frac{15}{128}\alpha}{\left(\frac{3}{5}\right)^{\frac{\alpha}{2}}}\right)\left(\frac{1}{4}\right)^{\alpha}.
 \end{aligned}
\end{equation}

 The lower bound $T_2$ in \cite[Thm. 3]{SXZWF} is a special case of our bound at $p_{1_{B}}=p_{2_{B}}=1$, 
\begin{align}\notag
T_{2}&=\left(\frac{63}{64}\right)^{\frac{\alpha}{2}}-\Gamma_B
=\left(\frac{63}{64}\right)^{\frac{\alpha}{2}}
-\left(1+\frac{1}{24}\alpha\right)\left(\frac{3}{4}\right)^{\alpha}\\ \label{e:chap4.1-ineq8}
&-\left(2^{\frac{\alpha}{2}}-1-\frac{1}{8}\alpha\right)\left(1+\frac{1}{16}\alpha\right)\left(\frac{\sqrt{2}}{4}\right)^{\alpha}
-\left(2^{\frac{\alpha}{2}}-1-\frac{1}{8}\alpha\right)^{2}\left(\frac{1}{4}\right)^{\alpha}
\end{align}

The lower bound given in \cite[Thm. 3]{JF} is
\begin{align*}
T_{3}=\left(\frac{63}{64}\right)^{\frac{\alpha}{2}}-\left(\frac{3}{4}\right)^{\alpha}-\left(2^{\frac{\alpha}{2}}-1\right)\left(\frac{\sqrt{2}}{4}\right)^{\alpha}-\left(2^{\frac{\alpha}{2}}-1\right)^{2}\left(\frac{1}{4}\right)^{\alpha}.
 \end{align*}

 The lower bound given in \cite[Thm. 2]{JFL} is
\begin{align*}
T_{4}=\left(\frac{63}{64}\right)^{\frac{\alpha}{2}}-\left(\frac{3}{4}\right)^{\alpha}-\frac{\alpha}{2}\left(\frac{\sqrt{2}}{4}\right)^{\alpha}-\left(\frac{\alpha}{2}\right)^{2}\left(\frac{1}{4}\right)^{\alpha}.
 \end{align*}

\begin{figure}[H] 
    \centering 
    \includegraphics[width=0.8\textwidth]{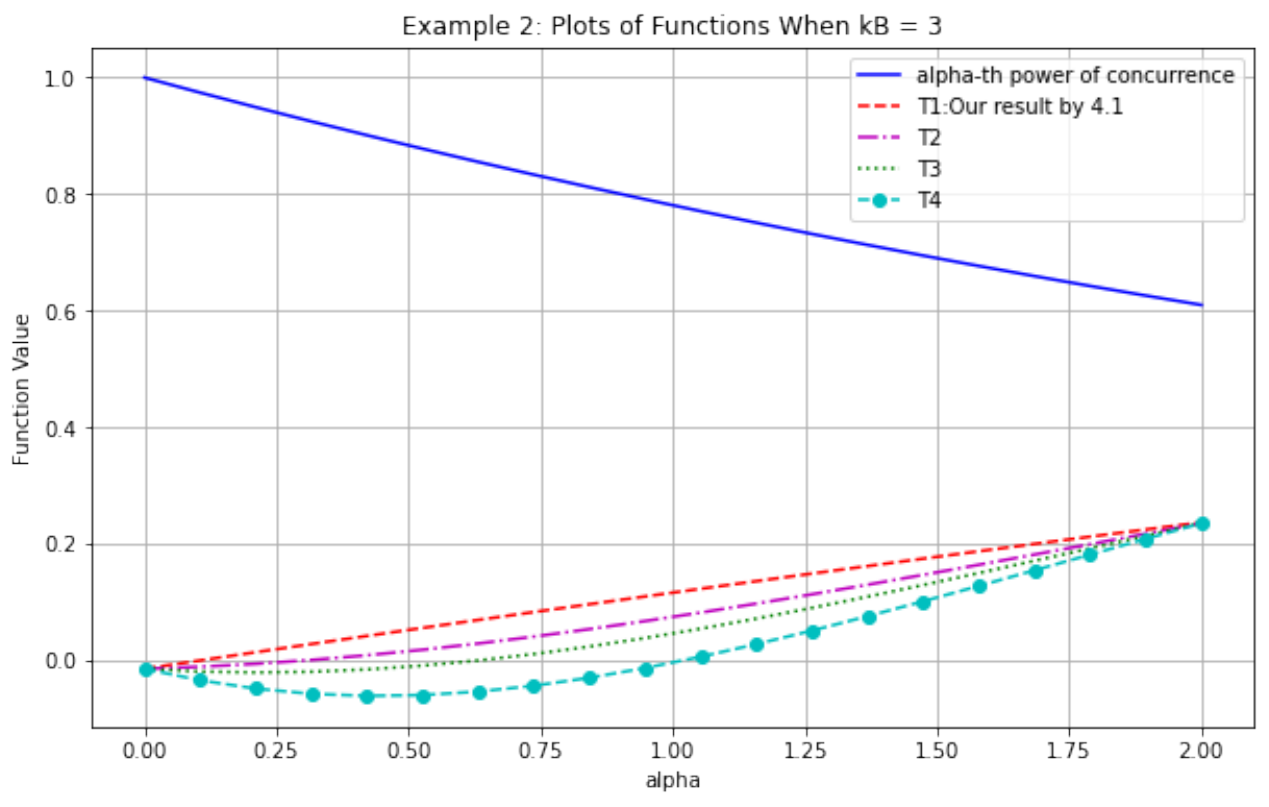} 
\end{figure}
Figure 2: The $x$-axis is $\alpha$. From top to bottom, the solid blue line is the $\alpha$th power of concurrence for $|W\rangle_{AB|C_{1}C_{2}}$. The dashed red line $T_1$, the dash-dot magenta line $T_2$, the dotted green line $T_3$ and the dashed cyan line with circle markers $T_4$ represent 
our bound \eqref{e:chap4.1-ineq1}, those from \cite{SXZWF}, \cite{JF}, and \cite{JFL} respectively.
The graph shows our lower bound $Z_1$ is the strongest.

\item (Polygamy, upper bounds). Consider the generalized polygamy for $C^{\alpha}(|W\rangle_{AB|C_{1}C_{2}})$.
By \eqref{e:chap4.1-ineq2}, our upper bound is
\begin{align*}
X_{1}&=\Theta_{B}+\left(1+\frac{25}{216}\alpha\right)\left(\frac{3}{4}\right)^{\alpha}+\left(\frac{\left(\frac{9}{5}\right)^{\frac{\alpha}{2}}-1-\frac{10}{81}\alpha}{\left(\frac{4}{5}\right)^{\frac{\alpha}{2}}}\right)\left(1+\frac{25}{256}\alpha\right)\left(\frac{3\sqrt{2}}{8}\right)^{\alpha}\\
&+\left(\frac{\left(\frac{9}{5}\right)^{\frac{\alpha}{2}}-1-\frac{10}{81}\alpha}{\left(\frac{4}{5}\right)^{\frac{\alpha}{2}}}\right)\left(\frac{\left(\frac{8}{5}\right)^{\frac{\alpha}{2}}-1-\frac{15}{128}\alpha}{\left(\frac{3}{5}\right)^{\frac{\alpha}{2}}}\right)\left(\frac{3}{8}\right)^{\alpha},
 \end{align*}
 where $\Theta_{B}$ is the same as \eqref{e:chap4.1-ineq7}. 

 The upper bound of \cite[Thm. 4]{SXZWF} is a special case of our bound at $p_{1_{A}}=p_{2_{A}}=1$:
 \begin{align*}
X_{2}&=\Gamma_B+\left(1+\frac{3}{32}\alpha\right)\left(\frac{3}{4}\right)^{\alpha}+\left(2^{\frac{\alpha}{2}}-1-\frac{1}{8}\alpha\right)\left(1+\frac{1}{16}\alpha\right)\left(\frac{3\sqrt{2}}{8}\right)^{\alpha}\\
&\qquad+\left(2^{\frac{\alpha}{2}}-1-\frac{1}{8}\alpha\right)^{2}\left(\frac{3}{8}\right)^{\alpha}.
 \end{align*}
where $\Gamma_{B}$ is given in \eqref{e:chap4.1-ineq8}.

 The upper bound $X_3$ in \cite[Thm. 4]{JF} is given by
\begin{align*}
X_{3}&=2\left(\frac{3}{4}\right)^{\alpha}+\left(2^{\frac{\alpha}{2}}-1\right)\left(\left(\frac{3\sqrt{2}}{8}\right)^{\alpha}+\left(\frac{\sqrt{2}}{4}\right)^{\alpha}\right)+\left(2^{\frac{\alpha}{2}}-1\right)^{2}\left(\left(\frac{3}{8}\right)^{\alpha}+\left(\frac{1}{4}\right)^{\alpha}\right)
\end{align*}

The upper bound given in \cite[Thm. 3]{JFL} is
\begin{align*}
X_{4}&=2\left(\frac{3}{4}\right)^{\alpha}+\frac{\alpha}{2}\left(\left(\frac{3\sqrt{2}}{8}\right)^{\alpha}+\left(\frac{\sqrt{2}}{4}\right)^{\alpha}\right)+\left(\frac{\alpha}{2}\right)^{2}\left(\left(\frac{3}{8}\right)^{\alpha}+\left(\frac{1}{4}\right)^{\alpha}\right)
\end{align*}
\end{enumerate}
 \begin{figure}[H] 
    \centering 
    \includegraphics[width=0.8\textwidth]{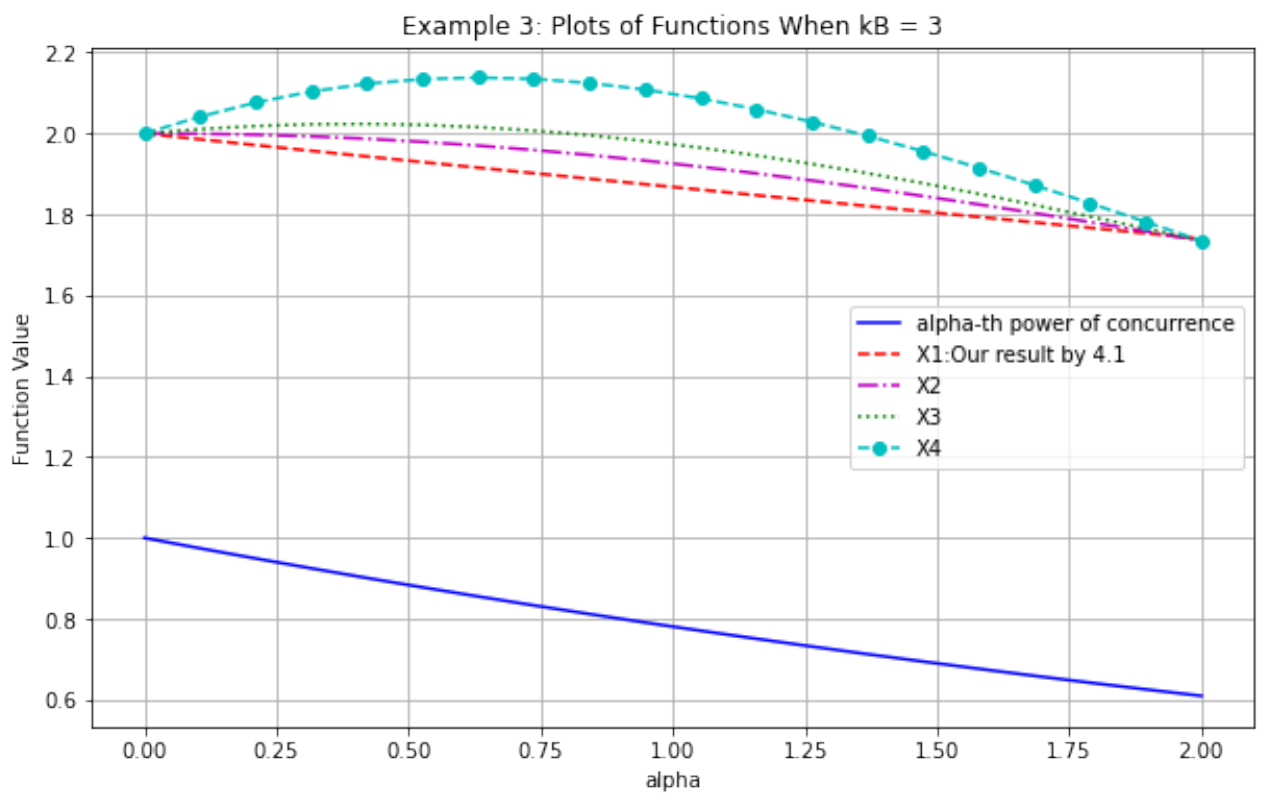} 
\end{figure}
Figure 3: The $x$-axis is $\alpha$. From bottom to top, the solid blue line is the $\alpha$th power of concurrence for $|W\rangle_{AB|C_{1}C_{2}}$, the dashed red line $X_1$ represents the upper bound from our result in \eqref{e:chap4.1-ineq2}, the Dash-dot magenta line $X_2$ represents the result in \cite{SXZWF}, the dotted green line $X_3$ represents the result in \cite{JF}, and the dashed cyan line with circle markers $X_4$ represents the result in \cite{JFL}. The graph shows our upper bound $X_1$ is the strongest.
\end{exmp}

\subsection{\textbf{The optimized bounds of $\mathcal{N}^{\alpha}(|\psi\rangle_{AB|C_{1}\cdots C_{N-2}})$ }} \label{sub:optimized bound-N}
Based on the relationship between negativity and concurrence, we can also obtain similar results for negativity.
Recall that the negativity of bipartite state $\rho_{A B}$ is defined by \cite{VW}: $N\left(\rho_{AB}\right)=\left\|\rho_{AB}^{T_{A}}\right\|-1$, where $T_{A}$ refers to the partial transposition 
and $\|X\|=\operatorname{Tr} \sqrt{X X^{\dagger}}$ is the trace norm of $X$. 

For a mixed state $\rho_{AB}$, the convex-roof extended negativity (CREN) is defined by 
$$N_{ c}\left(\rho_{A B}\right)=\min \sum_i p_i N\left(\left|\psi_i\right\rangle_{AB}\right)$$ 
where the minimum is taken over all possible pure state decompositions $\left\{p_i,\left|\psi_i\right\rangle_{AB}\right\}$ of $\rho_{AB}$. The convex-roof extended negativity of assistance (CRENoA) is then defined by $N_{a}\left(\rho_{AB}\right)=\max \sum_i p_i N\left(\left|\psi_i\right\rangle_{AB}\right)$, where the maximum is taken over all possible pure state decompositions $\left\{p_i,\left|\psi_i\right\rangle_{AB}\right\}$ of $\rho_{AB}$.

 For any bipartite pure state $|\psi\rangle_{AB}$ in a $d\times d$ quantum system with Schmidt rank $d$ (the total number of Schmidt numbers counting multiplicities), Ref. \cite{JF} shows that
\begin{equation}\label{e:chap4.2-ineq1}
\mathcal{N}(|\psi\rangle_{AB})\geq C(|\psi\rangle_{AB}).
\end{equation}
with equality holds when Schmidt rank is $d=2$. In the latter case, one obtains that
 \begin{equation}\label{e:chap4.2-ineq2}
\mathcal{N}_{c}(\rho_{AB})= C(\rho_{AB});\quad \mathcal{N}_{a}(\rho_{AB})= C_{a}(\rho_{AB}).
\end{equation}

 Combining \eqref{e:chap4.2-ineq1} and \eqref{e:chap4.2-ineq2}, an arbitrary $N$-qubit pure state $|\psi\rangle_{AB_{1}\cdots B_{N-1}}$ obeys the
 following monogamy and polygamy inequalities \cite{JF}:
\begin{equation}\label{e:chap4.2-ineq3}
\sum_{i=1}^{N-1} \mathcal{N}_{a}^{2}(\rho_{AB_{i}})\geq \mathcal{N}^{2}(|\psi\rangle_{A|B_{1}\cdots B_{N-1}})\geq \sum_{i=1}^{N-1} \mathcal{N}_{c}^{2}(\rho_{AB_{i}}),
\end{equation}


Meanwhile, by \cite{YCLZ}, for an arbitrary $N$-qubit pure state $|\psi\rangle_{AB_{1}\cdots B_{N-1}}$, we then have
\begin{equation}\label{e:chap4.2-ineq5}
\mathcal{N}^{2}(|\psi\rangle_{AB_{1}|B_{2}\cdots B_{N-1}})\leq \sqrt{\frac{r(r-1)}{2}}C(|\psi\rangle_{AB_{1}|B_{2}\cdots B_{N-1}}),
\end{equation}
 where $r$ is the Schmidt rank of the $N$-qubit pure state $|\psi\rangle_{AB_{1}\cdots B_{N-1}}$ under partition $AB_{1}$ and $B_{2}\cdots B_{N-1}$. For convenience, we will abbreviate: $N_{c(a)AB}=N_{c(a)}\left(\rho_{AB}\right)$.

Similar to concurrence, for any $N$-qubit state $|\psi\rangle_{AB|C_{1}\cdots C_{N-2}}$, $(N\geq4)$,
the $\alpha$th-generalized monogamy and generalized polygamy 
for negativity are straightforward 
conclusions by \eqref{e:chap4.1-ineq1}, \eqref{e:chap4.1-ineq2}, \eqref{e:chap4.2-ineq3}, 
and \eqref{e:chap4.2-ineq5}. Here $0\leq \alpha \leq 2$.
\begin{enumerate}[label={\bf Case \arabic*.}]
\item The lower bound of generalized monogamy for $\mathcal{N}^{\alpha}(|\psi\rangle_{AB|C_{1}\cdots C_{N-2}})$:
\begin{equation}\label{e:chap4.2-ineq6}
\begin{aligned}
\mathcal{N}^{\alpha}(|\psi\rangle_{AB|C_{1}\cdots B_{N-2}})\geq \max \left\{\left(\sum_{i=1}^{N-2}\mathcal{N}^{2}_{cAC_{i}}+\mathcal{N}^{2}_{cAB}\right)^{\frac{\alpha}{2}}-\Theta_{B}^{'},\left(\sum_{i=1}^{N-2}\mathcal{N}^{2}_{cBC_{i}}+\mathcal{N}^{2}_{cAB}\right)^{\frac{\alpha}{2}}-\Theta_{A}^{'}\right\}
 \end{aligned}
 \end{equation}

\item The upper bound of generalized polygamy for $\mathcal{N}^{\alpha}(|\psi\rangle_{AB|C_{1}\cdots C_{N-2}})$:
\begin{equation}\label{e:chap4.2-ineq7}
\begin{aligned}
\mathcal{N}^{\alpha}(|\psi\rangle_{AB|C_{1}\cdots C_{N-2}})\leq \left(\frac{r(r-1)}{2}\right)^{\frac{\alpha}{2}}(\Theta_{A}^{'}+\Theta_{B}^{'})
   \end{aligned}
\end{equation}
 \end{enumerate}
 where  $$\Theta_{j}^{'}=\sum_{l=1}^{k_{j}-1}\prod_{u=0}^{l-1}\Upsilon_{u_{j}}^{'}\Omega_{l_{j}}^{'}\mathcal{N}_{a}^{\alpha}(\rho_{jN_{l_{j}}})+\prod_{u=1}^{k_{j}-1}\Upsilon_{u_{j}}^{'}\mathcal{N}_{a}^{\alpha}(\rho_{jN_{k_{j}}}), $$
with $\Omega_{l_{j}}^{'}=1+\frac{\alpha}{2(1+p_{l_{j}})^{2}}\frac{\sum_{i=l+1}^{k_{j}}\mathcal{N}_{a}^{2}(\rho_{jN_{i_{j}}})}{\mathcal{N}_{a}^{2}(\rho_{jN_{l_{j}}})}, l=1,2,\cdots, k_{j}-1;$ $\Upsilon_{0_{j}}^{'}=1$ and $\Upsilon_{u_{j}}^{'}=\frac{(1+p_{u_{j}})^{\frac{\alpha}{2}}-1}{p_{u_{j}}^{\frac{\alpha}{2}}}-\frac{\alpha p_{u_{j}}}{2p_{u_{j}}^{\frac{\alpha}{2}}(1+p_{u_{j}})^{2}}$, $u=1,2,\cdots, k_{j}-1$, $2\leq k_{j} \leq N-1,$  $j=A,B.$

\bigskip
The following results are clear from the above discussion.
\begin{enumerate}[label= {\bf(\roman*)}]
\item The lower bound of generalized monogamy for $\mathcal{N}^{\alpha}(|\psi\rangle_{AB|C_{1}\cdots C_{N-2}})$ given in \cite[Thm. 6]{SXZWF} is a special case of our  \eqref{e:chap4.2-ineq6} with all $p_{i_{j}}=1,i=1,\cdots,k_{j}-1,j=A,B.$ The authors \cite{SXZWF} have shown that their results are stronger than \cite{JF,JFL}. Therefore our result \eqref{e:chap4.2-ineq6} is a stronger lower bound of generalized monogamy than \cite[Thm. 7]{JF}, \cite[Thm. 5]{JFL}.
\item The upper bound of generalized polygamy for $\mathcal{N}^{\alpha}(|\psi\rangle_{AB|C_{1}\cdots C_{N-2}})$ given in \cite[Thm. 7]{SXZWF} is a special case of our  \eqref{e:chap4.2-ineq7} with all $p_{i_{j}}=1,i=1,2,\cdots,k_{j}-1,j=A,B.$ The authors \cite{SXZWF} have proved that their results are stronger than \cite{JF,JFL}. Thus,
    we remark that our result  \eqref{e:chap4.2-ineq7} is the optimized upper bound of generalized polygamy relation than \cite[Thm. 8]{JF}, \cite[Thm. 6]{JFL}.
\end{enumerate}

\bigskip
Similarly, we can generalize the results in Subsection \ref{sub:optimized bound-C} under the partition $ABC_{1}$ and $C_{2}\cdots C_{N-2}$.
\subsection{\textbf{Stronger bounds of $C^{\alpha}(|\psi\rangle_{ABC_{1}|C_{2}\cdots C_{N-2}})$ }}

For $0\leq \alpha \leq 2$ and any $N$-qubit state $|\psi\rangle_{ABC_{1}|C_{2}\cdots C_{N-2}}$, $(N\geq5)$, the following results are direct by Subsection \ref{sub:optimized bound-C}.

\begin{enumerate}[label={\bf Case \arabic*.}]
\item The lower bound of generalized monogamy for $C^{\alpha}(|\psi\rangle_{ABC_{1}|C_{2}\cdots C_{N-2}})$:
\begin{equation}\label{e:chap4.3-ineq1}
\begin{aligned}
 &C^{\alpha}(|\psi\rangle_{ABC_{1}|C_{2}\cdots C_{N-2}})\\
 &\geq \max \left\{\left(\sum_{i=1}^{N-2}C^{2}_{AC_{i}}+C^{2}_{AB}\right)^{\frac{\alpha}{2}}-\Theta_{B}, \left(\sum_{i=1}^{N-2}C^{2}_{BC_{i}}+C^{2}_{AB}\right)^{\frac{\alpha}{2}}-\Theta_{A}\right\}
 -\Theta_{C_{1}}
 \end{aligned}
 \end{equation}
or
\begin{equation}\label{e:chap4.3-ineq2}
\begin{aligned}
 C^{\alpha}(|\psi\rangle_{ABC_{1}|C_{2}\cdots C_{N-2}})\geq \left(C^{2}_{AC_{1}}+C^{2}_{BC_{1}}+\sum_{i=2}^{N-2}C^{2}_{C_{1}C_{i}}\right)^{\frac{\alpha}{2}}-\Theta_{A}-\Theta_{B}
 \end{aligned}
 \end{equation}
\item The upper bound of generalized polygamy for $C^{\alpha}(|\psi\rangle_{ABC_{1}|C_{2}\cdots C_{N-2}})$:
\begin{equation}\label{e:chap4.3-ineq3}
\begin{aligned}
C^{\alpha}(|\psi\rangle_{ABC_{1}|C_{2}\cdots C_{N-2}})\leq \Theta_{A}+\Theta_{B}+\Theta_{C_{1}}
   \end{aligned}
\end{equation}
 \end{enumerate}
 where $\Theta_{j}$ $(j=A,B,C_{1})$ are defined similarly as \eqref{e:chap4.1-ineq1}, \eqref{e:chap4.1-ineq2}.

\begin{proof}
If $C^{2}(|\psi\rangle_{AB|C_{1}\cdots C_{N-2}}) \geq C^{2}(|\psi\rangle_{C_{1}|AB\cdots C_{N-2}})$, then

\begin{align*}
& C^{\alpha}(|\psi\rangle_{ABC_{1}|C_{2}\cdots C_{N-2}})=(2T(\rho_{ABC_{1}}))^{\frac{\alpha}{2}}\geq |2T(\rho_{AB})-2T(\rho_{C_{1}})|^{\frac{\alpha}{2}} \\
   & =| C^{2}(|\psi\rangle_{AB|C_{1}\cdots C_{N-2}})- C^{2}(|\psi\rangle_{C_{1}|ABC_{2}\cdots C_{N-2}})|^{\frac{\alpha}{2}}\quad \quad(\text{by Eqs. \eqref{e:Pure}, \eqref{e:entropy}, \eqref{e:property} })\\
   &\geq  C^{\alpha}(|\psi\rangle_{AB|C_{1}\cdots C_{N-2}})- C^{\alpha}(|\psi\rangle_{C_{1}|ABC_{2}\cdots C_{N-2}})\quad \quad(\text{by Lemma in \cite{JF}})\\
     &\geq \max \left\{\left(\sum_{i=1}^{N-2}C^{2}_{AC_{i}}+C^{2}_{AB}\right)^{\frac{\alpha}{2}}-\Theta_{B}, \left(\sum_{i=1}^{N-2}C^{2}_{BC_{i}}+C^{2}_{AB}\right)^{\frac{\alpha}{2}}-\Theta_{A}\right\} -\Theta_{C_{1}},
  \end{align*}
 where the last inequality is due to \eqref{e:chap3.1-ineq2}  and \eqref{e:chap4.1-ineq1}.

\bigskip
 Similarly, if $C^{2}(|\psi\rangle_{AB|C_{1}\cdots C_{N-2}}) \geq C^{2}(|\psi\rangle_{C_{1}|AB\cdots C_{N-2}})$, then
\begin{align*}
& C^{\alpha}(|\psi\rangle_{ABC_{1}|C_{2}\cdots C_{N-2}})\\
   &\geq  C^{\alpha}(|\psi\rangle_{C_{1}|ABC_{2}\cdots C_{N-2}})- C^{\alpha}(|\psi\rangle_{AB|C_{1}\cdots C_{N-2}})\quad \quad(\text{by Lemma in \cite{JF}})\\
     &\geq \left(C^{2}_{AC_{1}}+C^{2}_{BC_{1}}+\sum_{i=2}^{N-2}C^{2}_{C_{1}C_{i}}\right)^{\frac{\alpha}{2}}-\Theta_{A}-\Theta_{B}   \quad \quad(\text{by Eqs. \eqref{e:mono-N}, \eqref{e:chap4.1-ineq2}}).
   \end{align*}
On the other hand, we also have
\begin{align*}
& C^{\alpha}(|\psi\rangle_{ABC_{1}|C_{2}\cdots C_{N-2}})=(2T(\rho_{ABC_{1}}))^{\frac{\alpha}{2}}\leq |2T(\rho_{AB})+2T(\rho_{C_{1}})|^{\frac{\alpha}{2}} \\
   & =| C^{2}(|\psi\rangle_{AB|C_{1}\cdots C_{N-2}})+ C^{2}(|\psi\rangle_{C_{1}|ABC_{2}\cdots C_{N-2}})|^{\frac{\alpha}{2}}\quad \quad(\text{by Eqs. \eqref{e:Pure}, \eqref{e:entropy}, \eqref{e:property} })\\
   &\leq  C^{\alpha}(|\psi\rangle_{AB|C_{1}\cdots C_{N-2}})+ C^{\alpha}(|\psi\rangle_{C_{1}|ABC_{2}\cdots C_{N-2}})\quad \quad(\text{by Lemma in \cite{JF}})\\
     &\leq \Theta_{A}+\Theta_{B}+\Theta_{C_{1}} \quad \quad(\text{by Eqs. \eqref{e:chap3.1-ineq2}, \eqref{e:chap4.1-ineq2}})
  \end{align*}
\end{proof}

Thus, we have the following results.
\begin{enumerate}[label= {\bf(\roman*)}]
\item The lower bound of generalized monogamy for $C^{\alpha}(|\psi\rangle_{ABC_{1}|C_{2}\cdots C_{N-2}})$ given in \cite[Cor. 1]{SXZWF} is a special case of our  \eqref{e:chap4.3-ineq1} or  \eqref{e:chap4.3-ineq2} with $p_{i_{j}}=1,i=1,2,\cdots,k_{j}-1,j=A,B,C_{1}.$  We remark that  \eqref{e:chap4.3-ineq1} or  \eqref{e:chap4.3-ineq2} is the optimized lower bound of generalized monogamy relation than \cite[Eq. (15) or Eq. (16)]{JF}, \cite[Eq. (12) or Eq. (13)]{JFL}.
\item The upper bound of generalized polygamy for $C^{\alpha}(|\psi\rangle_{ABC_{1}|C_{2}\cdots C_{N-2}})$ given in \cite[Cor. 2]{SXZWF} is a special case of our  \eqref{e:chap4.3-ineq2} with $p_{i_{j}}=1,i=1,2,\cdots,k_{j}-1,j=A,B,C_{1}.$ We remark that  \eqref{e:chap4.3-ineq3} is the optimized upper bound of generalized polygamy relation than \cite[Eq. (17)]{JF}, \cite[Eq. (14)]{JFL}.
\end{enumerate}
For any $N$-qubit state $|\psi\rangle_{A_{1}\cdots A_{k} B_{1}\cdots B_{N-k}}$, if we consider the partition under $A_{1}\cdots A_{k}$ and $B_{1}\cdots B_{N-k}$, then the second Case can be generalized as follows.
 \begin{rem}
The upper bound of generalized polygamy for $C^{\alpha}(|\psi\rangle_{A_{1}\cdots A_{k}| B_{1}\cdots B_{N-k}})$:
\begin{equation}\label{e:chap4.3-ineq4}
\begin{aligned}
C^{\alpha}(|\psi\rangle_{A_{1}\cdots A_{k}| B_{1}\cdots B_{N-k}}) \leq \sum_{i=1}^{k} \Theta_{A_{i}},
\end{aligned}
\end{equation}
for $0\leq \alpha \leq 2$, where $\Theta_{j}$ $(j=A_{i},i=1,2,\cdots,k)$ are defined similarly as \eqref{e:chap4.1-ineq1} and \eqref{e:chap4.1-ineq2}.
\end{rem}

\section{\textbf{Conclusion}}
By introducing a family of inequalities indexed by parameter $p$ $(0<p\leq 1)$, we get a family of tighter polygamy inequalities of CoA. We illustrate in detail how the parameter helps us to obtain stronger polygamy relations in various cases. We then have derived the generalized monogamy and polygamy inequalities based on the $\alpha$th ($0\leq \alpha \leq 2$) power of concurrence for any $N$-qubit pure states $|\psi\rangle_{ABC_{1}\cdots C_{N-2}}$ under the partition $AB$ and $C_{1}\cdots C_{N-2}$ as well as under the partition $ABC_{1}$ and $C_{2}\cdots C_{N-2}$. Comparisons are given in detailed examples which show that our lower (resp. upper) bounds of the generalized monogamy (resp. polygamy) relations for the $\alpha$th ($0\leq \alpha \leq 2$) power of concurrence are indeed the {\it tightest} in both cases among recent studies.

\bigskip
\centerline{\bf Acknowledgments}

The research is supported in part by the Simons Foundation and
National Natural Science Foundation of China.

\bigskip

{\bf Data availability statement}. All data generated or
analyzed during this study are included in this published
article.

\end{document}